\documentclass[11pt]{article}
\usepackage{amsmath,amssymb,amsthm}
\usepackage{hyperref}
\usepackage{microtype}
\usepackage{booktabs}
\title{Proof-of-Continuity: A Temporal Model\\
for Authority Propagation in\\
Distributed Systems and AI Agents}
\author{Nicola Gallo}
\date{8 July 2026}

\newtheorem{definition}{Definition}
\newtheorem{lemma}{Lemma}
\newtheorem{theorem}{Theorem}
\newtheorem{corollary}{Corollary}

\begin{document}
\maketitle

\begin{abstract}
Proof-of-Possession authorization models derive authority from the possession
of artifacts such as tokens, credentials, or capabilities. This paper argues
that possession is insufficient for discrete execution chains, whether they
span multiple services or occur as separated steps within the same machine,
because it does not guarantee preservation of the causal relationship between the
origin of a request and the authority exercised at later steps. We introduce
Proof-of-Continuity, a minimal authority-propagation discipline for the
Provenance Identity Continuity (PIC) model, in which each execution step must be
causally linked to the previous step and may only propagate a non-expansive
subset of the authority received from the origin. It introduces Proof of
Relationship, a single-hop causal primitive whose transitive composition is
Proof-of-Continuity; these complement Proof-of-Possession rather than replace it.
Under this model, the confused deputy condition cannot be satisfied as valid
model behavior: any privilege exercised at a later step must already be present
in the origin authority context. This is directly relevant to distributed
systems and AI agents, where executors invoke tools and downstream services
while holding multiple authority sources, so that the same authority/causality
mismatch recurs across service boundaries. Under Proof-of-Continuity these
sources may be carried together but are never merged into a combined authority,
since each step is authorized only against the authority context of the lineage
that caused it.

This paper concerns authorization propagation rather than authentication:
identity and authentication mechanisms such as OIDC, verifiable credentials,
wallets, and workload identity remain complementary mechanisms for establishing
the origin, while
Proof-of-Continuity addresses how authority propagates after that origin exists.
\end{abstract}

\section{Model}

We formalize the PIC (Provenance Identity Continuity) Model as follows.

Let $P$ be a finite set of principals, $O$ a set of operations, and $R$ a set of resources.
Define a privilege as $(o, r) \in O \times R$.

Each principal $p$ has an associated privilege set:
\[
Priv(p) \subseteq O \times R.
\]

A principal is used here as an abstract permissioned entity: it may be a human
identity established through an identity provider, wallet, OIDC/OIDC4VC-style
flow~\cite{openidconnect2014,openid4vc2024}, or decentralized
identifier~\cite{didcore2022}; a workload or machine identity such as those
used in WIMSE/SPIFFE-style environments~\cite{wimse2024,spiffe2017}; a role; a
service account; or another authenticated permissioned entity. What matters for this model is
that a privilege set $Priv(p)$ is associated with it and that it can express an
intent within those privileges. When a permissioned entity $p$ expresses such an
intent, it selects a subset of $Priv(p)$; that subset becomes the origin
authority context $C_0 \subseteq Priv(p)$ for a new execution lineage. The
privileges in $Priv(p)$ may be expressed at the level of operations the principal
is entitled to \emph{cause}, not only those it performs directly; a downstream hop
may realize such an authority through a locally different operation vocabulary via
the policy translation $\mathcal{T}$ of Section~\ref{sec:poc}.

For example, suppose Bob holds the privileges $(\text{read},\text{Book})$ and
$(\text{write},\text{Book})$. When Bob expresses an intent he selects a subset ---
say only $(\text{write},\text{Book})$ --- and that subset becomes the origin
authority context $C_0$ of a new lineage. Two executions may both involve Bob and
the same selected privilege, yet need not be the same occurrence: each carries
authority propagated along its own lineage. One may be a valid continuation of an
intent Bob expressed within a lineage; another may be an unrelated execution
exercising the same apparent privilege outside the lineage that granted it. The identity and privilege coordinates may be identical; the
lineage coordinate distinguishes the two.

A request is a message $req$ passed between executors --- the entities that
perform execution steps --- and may contain a payload.

Execution is modeled as a causal chain of steps:
\[
\pi = \langle (s_0, C_0), (s_1, C_1), \dots, (s_n, C_n) \rangle
\]
where $s_0$ is the origin step and $C_i \subseteq O \times R$ is the authority
context carried at step $s_i$.
We use $s_i$ for the formal execution step at hop $i$; operationally, a hop is
the execution of one such step by an executor that receives, processes, and may
propagate a request. Each hop may be performed by a service, workload, function,
tool, or agent; the model requires continuity of authority, not propagation of
identity.

\begin{definition}[PIC Model]
A system enforces Provenance Identity Continuity if, for every execution chain
$\pi$, each transition preserves causal continuity and monotonic authority
restriction:
\[
C_{i+1} \subseteq C_i
\]
for all $i$, and every privilege exercised at hop $n$ is causally linked to the
origin via a verifiable chain.
\end{definition}

The single-hop causal link asserted at each transition is a \emph{Proof of
Relationship} ($PoR$): evidence that a step is a valid continuation of its
immediate predecessor within the same lineage. The PIC Model thus combines two
conditions --- causal linkage, witnessed hop-by-hop by $PoR$, and monotonic
restriction $C_{i+1}\subseteq C_i$ --- both introduced informally here and
defined formally in Section~\ref{sec:poc}, where continuity is built as the
transitive composition of these relationships. The theorem below states the
authority consequence of the second condition.

\begin{theorem}[PIC Safety]
If the PIC Model holds, no execution step can exercise a privilege not present
at the origin.
\end{theorem}

\begin{proof}
By construction, $C_{i+1} \subseteq C_i$ for all $i$, so $C_n \subseteq C_0$.
Thus, any $(o,r) \in C_n$ must also be in $C_0$.
Therefore, no privilege can be gained beyond the origin authority context.
\end{proof}

\section{Threat Model and Scope}

We consider execution chains in which an origin request causes one or more later
execution steps. Executors may be buggy, confused, or adversarially influenced
through their inputs or their execution context. An executor may hold multiple authority sources --- ambient
privileges, delegated tokens, bearer credentials, or capabilities --- and may
attempt to use one unrelated to the request being processed.

The adversary may influence requests, payloads, and parameters passed across the
chain, and may cause an executor to select the wrong authority source. The
adversary cannot forge the relationship evidence required by $PoR$, nor make a
valid continuity checker accept a transition that violates monotonicity. As
capability and token models assume their credentials unforgeable, we assume the
relationship evidence and the authority-transition mechanism unforgeable;
compromise of the validator, the cryptographic root of trust, or the initial
authority bootstrap is out of scope.

\section{Confused Deputy}

\begin{definition}[Confused Deputy]\label{def:confused}
A confused deputy occurs when there exist principals $U$ (user) and $D$ (deputy)
and a privilege $(o,r) \in O \times R$ such that:
\begin{enumerate}
    \item $(o,r) \notin Priv(U)$,
    \item $(o,r) \in Priv(D)$,
    \item $U$ sends a request $req$ to $D$,
    \item as a consequence of $req$, $D$ executes $(o,r)$.
\end{enumerate}
\end{definition}

This definition does not depend on implementation details,
only on the mismatch of authority and causality.

\paragraph{Classical example.}
A FORTRAN compiler \texttt{FORT}, installed in the privileged directory
\texttt{SYSX}, holds ambient authority to write statistics to
\texttt{(SYSX)STAT}. A user invokes the compiler and provides
\texttt{(SYSX)BILL} as the debugging output filename. The compiler opens the
target for output using its own authority over \texttt{SYSX}, thereby
overwriting \texttt{(SYSX)BILL}. The user never possessed this privilege, but
the deputy exercised it as a consequence of the user's request~\cite{hardy1988confused}.

\section{Proof-of-Possession (PoP)}

A token $t$ represents a set of privileges:
\[
Priv(t) \subseteq O \times R.
\]

\paragraph{PoP Semantics.}
Possession implies usability:
if a principal holds $t$, it may exercise all $(o,r)\in Priv(t)$.

PoP systems constrain authority by artifact possession, but do not, by
themselves, require that the exercised authority be a causal continuation of
the authority context of the request.

\subsection{Sender-Constrained and Holder-Bound Tokens}

Sender-constrained and holder-bound tokens --- for example
DPoP~\cite{fett2023dpop} and Token Binding~\cite{popov2018token} --- prove
\emph{who} may present or use an artifact by binding it to a key or transport
channel. They do not, by themselves, prove that the authority exercised through
that artifact is a continuation of the current execution lineage. Restricting who
may present a token constrains transport, not causality: a legitimate holder can
still, at some step, exercise a possessed authority that lies outside the lineage
of the request it is processing. The binding answers holder authenticity; the
lineage question --- whether this use continues the authority context that caused
it --- remains open.

\subsection{Authority Mixing}

Assume:
\begin{align}
&(\text{write},r) \notin Priv(U) \tag{H1}\\
&(\text{write},r) \in Priv(D) \tag{H2}
\end{align}

Here $U$ is the user and $D$ is a deputy.

In a possession-based model, a deputy may hold multiple authority sources:
its own ambient privileges, delegated tokens, bearer credentials, or other
artifacts. The model authorizes an operation by checking whether some possessed
authority source grants the requested privilege.

\begin{definition}[Selection function]
Let an executor's held authority be a family $\{A_j\}$ of sources --- ambient
privileges, delegated tokens, and bound or sender-constrained artifacts. A
\emph{selection function} $\sigma$ chooses, for a requested $(o,r)$, some source
$A_j$ with $(o,r) \in A_j$, independently of the request lineage $\ell$ --- the
causal lineage of the request being processed, formalized in
Section~\ref{sec:poc}.
\end{definition}

\begin{theorem}[PoP admits authority mixing]
In a possession-based model, if an executor may select authority from any
artifact or ambient privilege it holds independently of the request authority
context, then the confused deputy condition is admissible in the model.
\end{theorem}

\begin{proof}
Let a request originate from an authority context that does not contain
$(o,r)$. Let the executor hold another authority context, artifact, or ambient
privilege that does contain $(o,r)$.

In a PoP model, authorization depends on the possession of an artifact or
privilege granting $(o,r)$, not on whether that authority is a continuation of
the request authority context. Therefore the model permits an execution in
which the request causally triggers $(o,r)$ while $(o,r)$ is absent from the
request authority context. This is the confused deputy condition of
Definition~\ref{def:confused}.
\end{proof}

This applies to any model in which authorization is derived from possession
independently of the causal lineage, including JWTs, bearer credentials, and
capabilities treated as held artifacts when they may be applied outside the
execution that requested them. Holder binding and sender-constraint restrict
which principal may present $A_j$, not whether $\sigma$'s choice lies in $\ell$;
they therefore do not remove this admissibility. The gap is not possession itself
but the absence of a temporal binding between the authority used and the execution
that caused its use.

These are execution-level conditions: authority belonging to one lineage is drawn
into another, producing the confused deputy. Proof-of-Continuity removes the
admissibility at the level of the model --- in a valid PIC execution the
confused-deputy condition cannot arise (Section~\ref{sec:cdimposs}), because
authority is defined only as a continuation of the causing lineage.

\section{Proof-of-Continuity (PoC / PIC)}
\label{sec:poc}

Execution is a causal chain:
\[
s_0 \rightarrow s_1 \rightarrow \dots \rightarrow s_n.
\]

Let $C_i \subseteq O \times R$ denote the authority context carried by execution
step $s_i$.

\begin{definition}[Proof of Relationship]
$PoR(s_i, s_{i+1})$ holds iff $s_{i+1}$ is a valid causal continuation of $s_i$
within the same execution lineage. $PoR$ is a \emph{relational}, single-hop
property: it binds one step to its immediate predecessor. The reach back to the
origin is not carried by any single $PoR$ but obtained by composing consecutive
relationships, as the definition of $PoC$ below makes precise.
\end{definition}

\begin{definition}[Valid transition]
A transition from $s_i$ to $s_{i+1}$ is valid iff it is both causally related
and non-expansive:
\[
Valid(s_i,C_i,s_{i+1},C_{i+1})
\Longleftrightarrow
PoR(s_i,s_{i+1}) \land C_{i+1}\subseteq C_i.
\]
\end{definition}

The first condition proves causal relationship.
The second condition proves monotonic authority restriction.
Thus each hop transfers only:
\[
C_{i+1} \subseteq C_i. \tag{C1}
\]

\begin{definition}[Proof of Continuity]
For a chain $\pi = \langle s_0, \dots, s_n \rangle$, $PoC(\pi)$ holds iff every
adjacent transition is valid:
\[
PoC(\pi) \Longleftrightarrow
\bigwedge_{i=0}^{n-1} \big( PoR(s_i, s_{i+1}) \land C_{i+1} \subseteq C_i \big).
\]
\end{definition}

Thus continuity is the transitive composition of relationships: $PoR$ proves a
single link, $PoC$ proves the whole lineage. Relationship is local; continuity
is global.

\paragraph{Monotone propagation along a causal chain.}
Two things propagate together. First, authority attenuates monotonically: write
$C_{i+1} \preceq C_i$ for authority \emph{attenuation} --- every authority at a
hop is a restriction of one held at the previous hop, never an expansion. When
the context is a set of atomic privileges this specializes to
$C_{i+1} \subseteq C_i$, the form used in our examples. Although the core model
represents authority contexts as sets of atomic privileges, we use $\preceq$ when
referring to the general attenuation order; in the set-valued case used
throughout the formal examples, this order is $\subseteq$. Second, the hops are
bound into a single causal lineage by $PoR$: continuity is a property not of a
hop in isolation but of the \emph{chain}. Formally the lineage is a logical
causal chain; in an implementation it need not be carried whole --- it may be
represented by cryptographic commitments, references, hashes, signatures, or
summarized proofs that allow each hop to prove its place in the lineage without
transporting the entire history.

\paragraph{Embedding execution in the authority model.}
Proof-of-Continuity is defined precisely for this setting: the successor step need
not be known, selected, instantiated, or provisioned when the predecessor
completes its step. Execution is a sequence of causal steps in time, not of
positions fixed by topology --- the upstream executor acts at some time $t$, while
the downstream executor may come into existence only at a later time $t' > t$.
Pre-binding authority to a concrete holder, key, or channel of an unknown
successor is therefore not always possible.

What must hold across the gap is not the identity of the successor, nor a policy
attached to it, but that its authority continues the context that caused it. A
credential, workload-identity assertion, service-account credential, or bound
token may authenticate the holder or bind an artifact to a key, workload, or
channel. Unless the authorization decision also reads lineage-sensitive evidence,
however, the same holder attributes may appear across distinct executions: they
establish who acts, not which causal chain produced this request. Re-deriving authority at $i+1$
from a fresh assertion about the successor, rather than from its relationship to
the causing step, reintroduces the individuation gap of
Section~\ref{sec:impossibility} --- the structural condition behind the confused
deputy (Definition~\ref{def:confused}). Continuity avoids this because it is not
self-assertable: defined only relative to the predecessor, it can be
\emph{demonstrated} as a continuation but not asserted at $i+1$ in isolation.

A predecessor may instead emit a continuation later consumed by some successor:
\[
(s_i,C_i)
\;\xrightarrow{\;\mathrm{emit}\;}
\text{unknown next hop}
\;\xrightarrow{\;PoR,\; C_{i+1}\subseteq C_i\;}
(s_{i+1},C_{i+1}).
\]
What persists across the gap is valid continuity, not possession by a pre-existing
holder: a causal relationship to the predecessor and a non-expansive authority
context. Proof of Relationship witnesses this as a property of the execution
rather than of a holder; composed transitively under monotone attenuation
($C_{i+1} \preceq C_i$), it is Proof-of-Continuity, the propagation invariant of
the PIC model.

\paragraph{Heterogeneous operation spaces.}
The notation assumes a common privilege universe $O \times R$. In heterogeneous
systems different hops use different vocabularies; let $O_i \times R_i$ be the
universe at hop $i$. A transition is then governed by a policy translation
\[
\mathcal{T}_{i \to i+1} : \mathcal{P}(O_i \times R_i) \to
\mathcal{P}(O_{i+1} \times R_{i+1}),
\]
and monotonicity becomes $C_{i+1} \subseteq \mathcal{T}_{i \to i+1}(C_i)$: the
successor authority is no greater than what the predecessor permits under the
local translation. We treat $\mathcal{T}$ as fixed policy input; synthesizing or
verifying such translations is out of scope. Heterogeneity here is semantic
rather than merely administrative: different hops may describe authority through
different operation/resource vocabularies even within a single trust perimeter.
In heterogeneous systems, safety is therefore relative to the soundness of the
chosen policy translation $\mathcal{T}_{i \to i+1}$; an overly permissive
translation is a policy error rather than a violation of the continuity invariant
itself.

\begin{lemma}[Continuity implies origin binding]
If $PoC(\pi)$ holds, then $s_n$ is causally linked to $s_0$ through a verifiable
relational chain and $C_n \subseteq C_0$.
\end{lemma}

\begin{proof}
Each $PoR(s_i,s_{i+1})$ links consecutive steps; their composition links $s_n$
to $s_0$. Non-expansiveness at each hop gives $C_n \subseteq \dots \subseteq
C_0$.
\end{proof}

We state the scope explicitly: this paper models the continuity invariant; the
concrete (e.g.\ cryptographic) construction of $PoR$ is out of scope and is
realized by a companion enforcement architecture. As capability and token models
assume their credentials unforgeable, we assume the lineage is unforgeable.

\subsection{Authorization Rule}

The following rule is evaluated only for a valid continuity chain, i.e.\ a chain
in which every adjacent transition satisfies both $PoR(s_i,s_{i+1})$ and
$C_{i+1}\subseteq C_i$.

For the common privilege universe $O \times R$, $(o,r)$ is authorized at the
final hop iff:
\[
(o,r) \in \bigcap_{i=0}^{n} C_i.
\]

Since the chain is monotonic decreasing:
\[
C_n \subseteq C_{n-1} \subseteq \dots \subseteq C_0,
\]
we have:
\[
\bigcap_{i=0}^{n} C_i = C_n.
\]

Thus no hop may acquire authority not already present at the origin, and any
authority dropped at some hop $j$ is lost irreversibly: by monotonicity it cannot
reappear at any later hop.

\subsection{The Ontological Shift: Possession vs Continuity}

Possession- and capability-based models~\cite{dennis1966programming,miller2003capability}
answer a \emph{spatial}
question: does an actor hold authority over an object? Authority is a point in
$O \times R$. This model is atemporal: it has no axis on which to record the
execution that led to an action.

Proof-of-Relationship supplies the missing axis: it binds an execution step to
its causal predecessor. Proof-of-Continuity is built on it --- continuity is the
transitive composition of relationships --- and combines two conditions: monotone
restriction of privileges ($C_{i+1} \subseteq C_i$) and causal continuity, via
$PoR$, of the execution that carries them.

Authority therefore no longer lives in $O \times R$ alone but in
$O \times R \times \mathcal{L}$, where $\mathcal{L}$ is the causal lineage of the
execution. Concretely, $\mathcal{L}$ is the set of finite linear causal
lineages: a lineage is a sequence
$\ell = \langle s_0 \prec s_1 \prec \dots \prec s_n \rangle$, where $\prec$
denotes causal precedence, witnessed between adjacent steps by $PoR$, so each
event is an occurrence of a privilege within a particular lineage. Two actions identical in possession --- same principal, same token,
same privilege --- may be distinct in continuity, because they descend from
different causes. The same $(o,r)$ exercised as a valid continuation of an
authorizing request, and exercised by an executor drawing on an unrelated
authority it happens to hold, is indistinguishable under possession, yet the two
are distinct executions: one continues an authorizing lineage, the other
continues none.
Possession is the projection onto $O \times R$ that forgets the lineage axis.
Possession is the spatial component of authority; lineage is its temporal
component, and a possession-only model is therefore a projection of authority
that forgets time. This does not make possession-based models wrong; it makes
them incomplete for the multi-hop case, where authority continuation becomes a
temporal property. Thus the relevant event is not merely $(p,o,r)$, or identity
plus privilege, but an occurrence of a privilege within a lineage. When the same
principal and the same privilege occur in two different lineages, the identity
coordinate does not separate them; the lineage coordinate does.

In the projection language this is the \emph{temporal reading} of the confused
deputy of Definition~\ref{def:confused}: an action spatially valid but temporally invalid --- the
deputy holds $(o,r)$, yet $(o,r) \notin C_0$ for the request that caused it ---
a collision under the projection that forgets $\mathcal{L}$, two distinct
executions that possession cannot tell apart. Section~\ref{sec:openpassthrough}
gives the operational characterization (the \emph{open passthrough}), which
refines Definition~\ref{def:confused} with designation.

\paragraph{Relation to capabilities.}
Capabilities solve the spatial question and more: by fusing designation and
authority in the act of invocation, a capability binds authority to a single
causal step --- in our terms \emph{single-hop continuity}, a
Proof-of-Relationship --- so the classical confused deputy is ruled out, not
merely reduced. What a capability treated purely as a held artifact lacks is the
extension of that binding across a chain: applied outside the lineage that
requested it, it falls back to lineage-invariant possession.

Proof-of-Possession and Proof-of-Continuity are two authority-propagation
primitives. A capability invocation is a one-hop instance of the continuity
primitive, whereas a capability held and presented as an artifact rests on
possession; a capability system that carried continuity along the whole chain
would, in these terms, be PIC-compliant. PIC is thus a temporal refinement of
capability-style propagation, not a replacement --- a dimension possession alone
cannot encode. Both eliminate the confused deputy by construction rather than by
runtime vigilance, but over different spans. Within the object-capability model,
fusing designation and authority at the point of invocation removes the classical,
single-hop mismatch structurally: a deputy cannot apply authority to a target it
was not given authority for. The mismatch re-appears only when the capability
model is left --- under ambient authority, or when a capability is used as a held
artifact outside the occurrence that carried it, the capability-as-possession case
analysed above, where authority reverts to possession. PIC extends the same
by-construction guarantee across a multi-hop lineage: the confused deputy is not a
state a valid PIC execution can represent at any hop, since a privilege absent
from the origin authority context is not a valid later state at all.

\subsection{Projection: Lineage-Invariant Possession Cannot Individuate Occurrences}
\label{sec:impossibility}

We make the projection precise. An \emph{event} is an occurrence exercising a
privilege within a lineage,
\[
e = (o, r, \ell) \in E := O \times R \times \mathcal{L},
\]
and $p : E \to O \times R$, $p(o,r,\ell) = (o,r)$, is the projection that forgets
the lineage.

\begin{definition}[Possession policy]
An authorization policy $A : E \to \{0,1\}$ is \emph{possession-based} iff it is
invariant under lineage: it factors through $p$, i.e.\ there exists
$\bar{A} : O \times R \to \{0,1\}$ with $A = \bar{A} \circ p$. Equivalently,
$A(o,r,\ell) = A(o,r,\ell')$ for all $\ell, \ell' \in \mathcal{L}$: the verdict
does not depend on which lineage produced the event.
\end{definition}

\begin{theorem}[Lineage-invariant policies cannot individuate occurrences]
Let $e = (o,r,\ell)$ and $e' = (o,r,\ell')$ be two events with the same privilege
but distinct lineage, $p(e) = p(e')$ and $\ell \neq \ell'$. Then every
possession-based policy satisfies $A(e) = A(e')$.
\end{theorem}

\begin{proof}
$A = \bar{A} \circ p$ and $p(e) = p(e')$, hence
$A(e) = \bar{A}(p(e)) = \bar{A}(p(e')) = A(e')$.
\end{proof}

The confused deputy is an instance of a forbidden pair:
events sharing $(o,r)$ but differing in lineage, one of which must be authorized
and the other denied. One lineage originates from a context holding $(o,r)$; the
other exercises $(o,r)$ though it never continued a context that held it. By the
theorem, \emph{no} lineage-invariant possession-based policy can separate them.

\begin{theorem}[Any resolution reintroduces continuity]\label{thm:reintro}
A policy $A$ authorizes $e$ and denies $e'$, with $p(e) = p(e')$ and
$\ell \neq \ell'$, only if $A$ is not lineage-invariant; that is, $A$ must read
some function $g : \mathcal{L} \to X$ of the lineage with $g(\ell) \neq g(\ell')$.
\end{theorem}

\begin{proof}
If $A$ were lineage-invariant it would factor through $p$ and give
$A(e) = A(e')$ by the previous theorem, contradicting $A(e) \neq A(e')$. Hence
$A$ depends on the lineage coordinate: some $g(\ell)$ distinguishes $e$ from
$e'$.
\end{proof}

Such a $g$ is exactly a continuity mechanism --- a proof-of-relationship extract
of the lineage. The mitigations deployed within possession-based systems
(nonces, idempotency keys, timestamps, audience and context binding) are
concrete instances of $g$: they resolve the confused deputy only by
reintroducing lineage information --- that is, continuity --- into the decision.
Whether the lineage is carried alongside the request or embedded inside the
artifact is immaterial: a token that encodes a nonce or a causal context makes
$A$ depend on $\ell$ and is therefore continuity-aware by definition, not
possession-based. The distinction is not where the information lives but whether
the decision reads it. The limitation is not an engineering failure of possession-based systems; it is
a modeling boundary of lineage-invariant authorization: by the first theorem it
is structurally unable to individuate occurrences, and every working fix is
lineage-sensitive --- continuity-aware --- rather than lineage-invariant.

\paragraph{Why the possession model is representative.}
In the abstraction considered here, a bearer token used without explicit context
or lineage binding is lineage-invariant: the decision reads the artifact and the
requested privilege, not the causal lineage of the use. A JWT, for instance,
asserts that ``the holder may read,'' not the causal chain that led to it. Such systems often work not because they
encode continuity explicitly, but because a perimeter supplies part of it
implicitly ---
mutual TLS, network boundaries, and transport-inferred delegation often stand in
for the missing lineage. Once the perimeter dissolves --- zero-trust deployments,
multi-organisation chains on the public internet, and autonomous AI agents
calling tools and services --- lineage-invariance becomes visible as a
vulnerability at the second hop. The definition of a possession policy above
therefore describes deployed systems, not a caricature of them. A trust perimeter
may hide the temporal gap by making authority and causality appear co-located,
but the model does not depend on crossing such a perimeter: the mismatch can
arise whenever authority is selected independently of the execution that caused
its use.

\subsection{A Trade-off Theorem for Authority Propagation}

The projection of Section~\ref{sec:impossibility} is the technical basis for a
trade-off among three properties an authority-propagation system may be asked to
satisfy at once.

\begin{definition}[Lineage-invariant authorization]
Authorization of an event $e = (o,r,\ell)$ depends only on the possessed
authority for $(o,r)$, not on the lineage $\ell$ that caused the event.
\end{definition}

\begin{definition}[Authority-mixing-capable delegation]
An executor may hold both a request-derived authority context $C_0$ and an
independent authority source $C_D$, with some $(o,r)$ such that
$(o,r) \notin C_0$ and $(o,r) \in C_D$.
\end{definition}

\begin{definition}[Confused-deputy safety]
No request whose origin context $C_0$ lacks $(o,r)$ can cause a later executor to
exercise $(o,r)$ using authority not derived from that request.
\end{definition}

\begin{theorem}[Possession--delegation--safety trade-off]
No authority-propagation system can simultaneously satisfy lineage-invariant
authorization, authority-mixing-capable delegation, and confused-deputy safety.
\end{theorem}

\begin{proof}
Suppose all three hold. By authority-mixing-capable delegation there is an
executor $D$ with a request-derived context $C_0$ and an independent source
$C_D$, and an $(o,r)$ with $(o,r) \notin C_0$ and $(o,r) \in C_D$. Let a request
cause $D$ to process a target requiring $(o,r)$. Lineage-invariant authorization
reads that $D$ possesses $C_D$ granting $(o,r)$; but by the projection theorem of
Section~\ref{sec:impossibility}, events sharing $(o,r)$ yet differing in lineage
receive the same judgment, so the policy cannot condition on whether $C_D$
belongs to the lineage of this request. The system therefore admits an execution
in which the request causes $D$ to exercise $(o,r)$ while $(o,r) \notin C_0$,
violating confused-deputy safety --- a contradiction.
\end{proof}

The theorem is not an attack on possession: it shows that lineage-invariant
possession is insufficient for multi-hop authority propagation precisely when
authority mixing and confused-deputy safety are both required. Since the three
cannot hold together, a system must relinquish at least one. For systems that
require both practical authority mixing and confused-deputy safety, the viable
choice is to relinquish lineage-invariance --- to make authorization
lineage-sensitive. Forbidding authority-mixing-capable delegation would instead
bar executors from holding authority of their own alongside request-derived
authority --- impractical in real deployments, where a service legitimately
retains its own privileges. Accepting confused-deputy executions abandons the
safety property at issue. Proof-of-Continuity takes this
option: it does not remove possession but restores the missing temporal
coordinate, so that authorization can read whether each authority use is a valid
continuation of the context that caused it. The contribution is not to reject
possession but to show that possession is the spatial projection of a richer
authority relation; the missing coordinate is temporal lineage.

\begin{table}[ht]
\centering
\small
\begin{tabular}{@{}p{2.9cm}p{5.0cm}p{3.1cm}@{}}
\toprule
Property relinquished & Consequence & Outcome \\
\midrule
Authority-mixing delegation & executors may not hold authority independent of the request & rejected (impractical) \\
Confused-deputy safety & confused-deputy executions become admissible & rejected (unsafe) \\
Lineage-invariance & authorization becomes lineage-sensitive & adopted (Proof-of-Continuity) \\
\bottomrule
\end{tabular}
\caption{The three properties cannot hold together, so a system must relinquish
at least one. For systems that require practical authority mixing and
confused-deputy safety, the viable choice is to relinquish lineage-invariance ---
the option adopted by Proof-of-Continuity.}
\label{tab:tradeoff}
\end{table}

\begin{corollary}[Continuity restores confused-deputy safety]
Relinquishing lineage-invariance regains safety: under non-expansive propagation
$C_{i+1} \subseteq C_i$, so $C_k \subseteq C_0$, and an $(o,r) \notin C_0$ cannot
appear at any later hop. This is the impossibility established in
Section~\ref{sec:cdimposs}.
\end{corollary}

\section{Safety Property}
\label{sec:openpassthrough}

\begin{definition}[Origin-bounded authority]
A model enforces origin-bounded authority if every executable privilege at hop
$n$ was originally granted at hop $0$.
\end{definition}

We refine Definition~\ref{def:confused} to separate legitimate delegation from
the confused deputy. The distinguishing variable is not the authority source but whether the
action's target is designated by the request and whether the requester was
entitled to bring it about.

\begin{definition}[Authority to cause]
The origin context $C_0$ records the authority the requester is entitled to
\emph{cause}, not only to perform directly. A later action belongs to the
requester's lineage iff it realizes an authority in $C_0$; otherwise it is the
deputy's own action, rooted in a distinct origin (the deputy itself) and
unconstrained by $C_0$.
\end{definition}

\begin{definition}[Open passthrough]
A deputy is an \emph{open passthrough} when, as a consequence of a request, it
applies authority to a target designated by that request for an action the
requester was not entitled to cause: $(o,r) \notin C_0$ yet $(o,r)$ executes
because of the request. This is exactly the confused deputy.
\end{definition}

These separate three cases. (i) A deputy's own housekeeping --- the compiler
writing its statistics file --- is rooted in the deputy as origin and is
unconstrained by $C_0$. (ii) Legitimate delegation --- a payment gateway moving
funds because the user was entitled to cause the payment --- realizes an
authority in $C_0$ and is a valid continuation, even though the user does not
directly hold the banking authority. (iii) The confused deputy --- writing a
user-designated file the user could not cause --- is an open passthrough, the one
case the model forbids. The relevant question is not whether the origin could
perform the downstream operation directly, but whether the origin was authorized
to \emph{cause} it within this lineage. When successive services use different
operation spaces, this is mediated by the policy translation $\mathcal{T}$
introduced above, whose synthesis and verification are left to future work.

\paragraph{Example.}
If $C_0 = \{(\text{convert},r)\}$, then $(\text{write},r)$ can never be authorized
within the same execution chain: $(\text{write},r)\notin C_0$ and
$C_{i+1}\subseteq C_i$ at every transition, so $(\text{write},r)\notin C_n$.

\section{Impossibility of the Confused Deputy}
\label{sec:cdimposs}

The confused deputy condition requires an authority mismatch:
a privilege absent from the authority context of the request is nevertheless
exercised as a consequence of that request.

This is the paper's second impossibility result. Section~\ref{sec:impossibility}
shows that possession \emph{cannot} individuate occurrences; here we show the dual: under
continuity the confused deputy is \emph{not a state the model can represent}.
The condition is unsatisfiable within any valid PIC execution --- an impossibility
by construction, at the level of valid model behavior.

\begin{theorem}[Confused deputy is impossible under PIC]
In any valid PIC execution --- one whose adjacent transitions satisfy $PoR$ and
monotonicity $C_{i+1}\subseteq C_i$, so that authority is origin-bounded
($C_n \subseteq C_0$) --- the confused deputy conditions cannot be jointly
satisfied as valid behavior within the model.
\end{theorem}

\begin{proof}
By $PoR$ the hops form a single lineage rooted at $s_0$, and by monotonicity
$C_k \subseteq C_0$ for all $k$; hence $(o,r)\notin C_0$ implies $(o,r)\notin C_k$.
Authority the deputy holds outside this lineage has a distinct origin
(Section~\ref{sec:openpassthrough}) and is not exercised as part of the request;
within the request's lineage the privilege mismatch a confused deputy requires is
therefore unsatisfiable.
\end{proof}

\paragraph{Interpretation.}
Authority in PIC is not something that a subject merely \emph{has};
it is a continuity property of the execution chain.
Since no hop may introduce or reinterpret authority absent at the origin,
a confused-deputy execution may still be attempted physically, but it cannot be
accepted as valid behavior within the PIC model.

\section{Discussion}

PoP systems intentionally derive authorization from artifact control, which is
often useful and necessary for local authorization decisions. In multi-hop
execution, however, artifact control alone does not determine whether a later use
is a valid continuation of the authority lineage that caused it. PIC/PoC systems
instead propagate only non-expansive subsets of authority along a causally
verified execution chain.

Proof of Possession proves artifact control.
Proof of Relationship proves execution relationship.
Proof of Continuity proves authority continuation.
Proof-of-Possession establishes control over an authority artifact at a point in
execution; it does not by itself establish that a later use is a valid
continuation of the authority lineage that caused it, which is the temporal
condition Proof-of-Continuity adds. Proof-of-Possession can serve as one
mechanism for constructing a Proof of Relationship --- possession of a
predecessor-issued continuation artifact can witness a single hop --- but it is
not the only one.

\paragraph{Authority is continuity.}
Under PIC, authority is not a static property an entity merely \emph{has} but a
continuity property of an execution: a privilege is authorized at a step iff it is
a non-expansive continuation of the context that caused it. In this sense
authority \emph{is} continuous --- it exists only as the ongoing continuation of
an origin, not as a possession detached from its lineage.

\paragraph{Design space.}
The temporal binding may be realized in two ways: carried \emph{inside} the
authority artifact itself, as attenuated and context-bound capabilities, or
maintained \emph{alongside} it as an explicit continuity proof. Either can be
continuity-aware: by Theorem~\ref{thm:reintro} what matters is not where the
binding lives but whether the authorization decision reads it. An artifact whose
binding is ignored --- presented outside the occurrence that carried it ---
reverts instead to lineage-invariant possession. PIC takes the latter position,
treating continuity as a first-class property of the execution chain; the two
approaches are complementary rather than exclusive.

The resulting security invariant is simple: no downstream service can perform,
as part of a valid execution chain, an operation that was not authorized at the
origin. This is a modeling refinement rather than a criticism of existing
mechanisms: possession, capabilities, bearer tokens, and proof-of-possession
bindings remain useful local authorization tools, while PoC makes explicit the
additional temporal condition needed when authority is propagated across an
execution lineage.

\paragraph{Authority as a distributed transaction.}
A multi-hop authority chain is a distributed transaction over authority rather
than over data. A possessed artifact enables a local authorization decision at the
point of use --- often necessary for deployability, performance, and loose
coupling --- but revocation, delegation state, scope changes, and lineage are
global authority facts. Local authorization and global authority consistency
diverge when these facts are not represented at the point of use. PoC addresses
this by carrying lineage explicitly and requiring non-expansive authority
propagation across the chain.

Continuity is therefore not identity propagation. Identity providers, wallets,
verifiable credentials~\cite{vcdm2022}, OIDC/OIDC4VC-style
flows~\cite{openidconnect2014,openid4vc2024}, and WIMSE/SPIFFE-style workload
identity~\cite{wimse2024,spiffe2017} remain important
for authentication, credential presentation, audit, accountability, and origin
bootstrap; but an identifier repeated at later hops does not by itself prove that
the authority being exercised belongs to the same execution lineage. A lineage
begins when a permissioned entity expresses an intent within its privileges,
creating an origin authority context $C_0$; subsequent hops need not carry or
reinterpret the origin identity, but must preserve a non-expansive authority
context causally linked to that origin.

\paragraph{Application to AI agents.}
For authority propagation, AI agents are a special case of autonomous distributed
execution; we name them because the temporal gap is most acute there. An agent
runs a chain of steps --- planning, tool invocations, downstream calls --- while
holding multiple authority sources at once: its own service credentials, delegated
user tokens, and per-tool scopes. The classical local-filename designation
reappears as an agent-influenced parameter --- a tenant id, document id, role identifier,
queue message, or request target. When the caller influences this target and the
agent satisfies it using authority from a different context, the result is the
confused deputy problem in distributed form. PIC binds each tool call and downstream request to
the causal lineage that triggered it, so authority applied outside the execution
that caused the action is not a valid continuation. An agent asked to summarize a document
carrying an indirect prompt injection --- instructions to delete files or
exfiltrate secrets --- may attempt those actions, but under PoC they are valid
only if the origin authority context holds the corresponding privileges: if it
authorizes only $(\text{summarize},d)$, then $(\text{delete},r)$ or
$(\text{exfiltrate},r)$ cannot be valid continuations.

\section*{Related Work}

The confused deputy problem was introduced by Hardy~\cite{hardy1988confused}.
Capability-based access control~\cite{dennis1966programming} and the analysis of
the confused deputy under capabilities~\cite{miller2003capability} address the spatial
dimension of authority --- whether an actor holds authority over an object. This
paper isolates the orthogonal temporal dimension: whether the authority used
belongs to the execution that caused the action. We take the confused deputy not
as the specific operating-system compiler scenario in which it was first
described, but as the general authority/causality mismatch that scenario
exemplifies, and recast it for discrete multi-step execution --- the setting of
distributed systems, service chains, and AI agents rather than a single local
invocation. Framed this way, the confused deputy is one instance of a class of
authority-propagation vulnerabilities that share the same mismatch; under
Proof-of-Continuity that mismatch cannot arise as valid model behavior, since a
privilege absent from the origin authority context cannot be exercised later in
the same lineage. The paper compares possession-based authority propagation with
provenance-continuity-based authority propagation on exactly this axis.

The closest prior work is history-based access control~\cite{abadi2003history},
where the authority available to a computation depends on the history of
what has executed. We share the premise that authority is a function of
execution rather than of static possession. Unlike a retrospective check over
past callers (as in stack inspection), PoC is a forward propagation discipline:
each hop must be causally related to its predecessor and must not expand the
authority it received, including across asynchronous and cross-boundary steps.
The contribution here is orthogonal
and more specific: we frame the dependence as a \emph{dimension} $\mathcal{L}$
orthogonal to the privilege $(o,r)$, prove that possession (any policy invariant
under $\mathcal{L}$) cannot individuate occurrences, and show that the confused
deputy is a collision under the projection that forgets
$\mathcal{L}$ --- so that possession-based mitigations that distinguish such
cases must become lineage-sensitive, hence continuity-aware. The
object-capability literature~\cite{miller2006robust} analyses the confused
deputy as a separation of designation from authority, which capabilities
reunify; this is the spatial resolution. Our result adds the temporal axis: in the abstraction of this paper, a capability
treated only as a portable held artifact, without occurrence or lineage binding,
falls into the lineage-invariant projection and cannot distinguish two
occurrences that differ only in lineage.

Provenance-based access control~\cite{park2012provenance} uses the recorded
provenance of objects as an input to authorization decisions; PIC instead makes
the causal lineage of the \emph{execution} itself the carrier of a non-expansive
authority context. Decentralized information-flow
control~\cite{myers1997decentralized} propagates labels on data under monotone
restriction; Proof-of-Continuity propagates authority along an execution lineage
under the same monotonicity discipline, but the propagated object differs ---
authority rather than information-flow labels.

\paragraph{Bearer tokens, macaroons, and proof-of-possession bindings.}
OAuth 2.0 includes bearer-token usage as a common authorization
pattern~\cite{hardt2012oauth,jones2012bearer}. A bearer token used without
explicit context or lineage binding supports a decision that reads the token and
$(o,r)$ but not the causal lineage of the use: it is lineage-invariant, and thus
subject to the first theorem. Macaroons~\cite{birgisson2014macaroons} add
\emph{caveats} that only attenuate and chain cryptographically --- precisely our
monotone attenuation $C_{i+1}\preceq C_i$ --- yet a macaroon remains a bearer
credential: absent a caveat binding the request context, it can be presented
outside the lineage that produced it. DPoP~\cite{fett2023dpop} and Token
Binding~\cite{popov2018token} bind a token to a key or transport channel; in the
terms of Theorem~\ref{thm:reintro} such bindings make the decision read
additional occurrence context, and are continuity-aware in a local sense. Each
reduces cross-context misuse to the extent that a caveat or binding makes the
decision read $\ell$; none by itself establishes global Proof-of-Continuity, since
the decision still need not read the full causal lineage from the origin across a
multi-hop chain. PIC makes that lineage dimension global and explicit.

The PIC model is developed independently in~\cite{gallo2025pic}, where Provenance
Identity Continuity is introduced as a general model for distributed execution
systems. The present paper extends an earlier version~\cite{gallo2025apm},
refining it with the relationship/continuity distinction and the spatial/temporal
analysis of authority.

\section{Limitations and Future Work}

This paper presents a minimal linear model of Proof-of-Continuity. Several of the
directions below are not defects of the model but properties to be developed
within the same continuity principle; we distinguish genuine scope limits from
natural extensions. First, revocation is not treated as a
retroactive operation on an already issued lineage: a practical system must
decide whether revocation blocks only future transitions, invalidates
outstanding continuations, or triggers lineage-wide cancellation. Second,
composition: the model is linear, and forks, joins, fan-out, fan-in, and
DAG-shaped executions are a natural generalization of $\mathcal{L}$ to a partial
order --- an extension of the continuity principle to composed lineages rather
than a limitation of it. Third, the construction of
$PoR$ is abstract: its concrete construction is an implementation concern
separate from the continuity model. The relationship evidence it requires can be
derived from mechanisms such as trusted execution environments, or
verifiable credentials, with implementation-dependent security and performance;
the model is agnostic to the choice.
Fourth, lineage validation is an implementation concern rather than a limitation
of the model: practical realizations need not revalidate an entire history at
every hop, and may use checkpoints, summarized proofs, compact encodings, or
accumulator-style constructions to keep validation cost bounded.
Fifth, lineage origination: the model constrains authority \emph{within} a lineage
but does not itself govern when an executor may legitimately open a new lineage
rooted in its own privileges. An adversarially influenced executor could
misclassify a request-caused action as self-originated, performing it outside the
victim's chain without expanding any $C_i$. Distinguishing request-caused actions
from executor-originated ones is a responsibility of the $PoR$ construction and the
enforcement architecture, left to future work.
Finally, heterogeneous systems require policy translations $\mathcal{T}$ between
local vocabularies, whose synthesis, verification, and governance we leave to
future work.

\section{Conclusion}

This paper introduced Proof-of-Continuity as a minimal temporal model for
authority propagation. The core distinction is between possession, which proves
control over an authority artifact; relationship, which proves causal linkage
between adjacent execution steps; and continuity, which proves that authority has
propagated without expansion across the whole lineage. The main result is that
lineage-invariant possession cannot distinguish occurrences that share the same
privilege but arise from different causes; conversely, under Proof-of-Continuity a
privilege absent from the origin authority context cannot appear later in the
same valid execution chain. The confused deputy condition is therefore not
accepted as valid model behavior --- it is precisely the authority/causality
mismatch that continuity rules out. More broadly, PoC reframes authority
propagation as a temporal property of an execution lineage rather than merely a
static property of artifact possession.

\paragraph{Note on the name PIC.}
In Provenance Identity Continuity, the term Identity refers to the security
identity or execution identity of the permissioned entity from which authority is
created or anchored at the origin. It does not mean that identity is propagated
as the authorization primitive at every hop. The model instead shifts the
authorization burden from repeatedly interpreting identity to verifying execution
continuity: a causal lineage carrying a non-expansive authority context.
Provenance denotes the causal origin and lineage of the execution, Identity
denotes the authenticated permissioned origin, and Continuity denotes the
preservation of authority across later hops without expansion. Identity and
identifiers therefore remain essential for authentication, credential
presentation, audit, accountability, and origin bootstrap; PoC concerns the
authorization continuity that must hold after that origin has been established.

\section*{Acknowledgments}

This paper follows an earlier version of the PIC Model made available on
Zenodo in December 2025~\cite{gallo2025pic}. The author thanks Alan H. Karp,
whose experience was valuable in stress-testing the concepts of
Proof-of-Continuity and Proof of Relationship. The concepts presented here are
the author's own, and this acknowledgment does not imply his endorsement of the
paper. The author also used automated language assistance tools for editorial
refinement, LaTeX formatting, and notation polishing. All conceptual
contributions, models, and proofs are the author's own.

\end{document}